\documentclass[reqno, 12pt]{amsart}

\usepackage{amsmath,amssymb,amsrefs,upgreek}
\usepackage[protrusion=true]{microtype}
\usepackage[T1]{fontenc}
\usepackage[english]{babel}
\usepackage[margin=1in]{geometry}
\usepackage[onehalfspacing]{setspace}			
\usepackage{graphics}			 				
\usepackage{pinlabel}							
\usepackage[bookmarks=false]{hyperref}		
\numberwithin{equation}{section}				
\usepackage[labelformat=empty]{caption}		
\usepackage[nameinlink,noabbrev]{cleveref}	
\usepackage[all]{hypcap}						
\usepackage{autonum}

\newtheorem{main}{Main Theorem}
\newtheorem{thm}{Theorem}[section]
\newtheorem{lem}[thm]{Lemma}
\newtheorem{cor}[thm]{Corollary}
\newtheorem{rem}[thm]{Remark}

\crefname{thm}{Theorem}{Theorems}			
\creflabelformat{thm}{#2{#1}#3}				
\crefname{ineq}{inequality}{inequalities}	
\creflabelformat{ineq}{#2{\upshape(#1)}#3}	

\makeatletter
\renewcommand*{\eqref}[1]{\hyperref[{#1}]{\textup{\tagform@{\ref*{#1}}}}}	
\makeatother

\let\originalleft\left
\let\originalright\right
\renewcommand{\left}{\mathopen{}\mathclose\bgroup\originalleft}
\renewcommand{\right}{\aftergroup\egroup\originalright}

\def\rl{\mathbb{R}}
\def\N{\mathbb{N}}
\def\Z{\mathbb{Z}}

\def\C{\mathcal{C}}
\def\E{\mathcal{E}}
\def\G{\mathcal{G}}
\def\M{\mathcal{M}}
\def\S{\mathcal{S}}

\def\Aut{\mathrm{Aut}}
\def\Im{\mathrm{Im}}

\def\U{\mathrm{U}}

\def\Ar{\mathrm{Area}}						

\title{Irreducible Ginzburg--Landau fields in dimension 2}
\author{\'Akos Nagy}
\date{\today}
\address[\'Akos Nagy]{Department of Mathematics, Michigan State University, East Lansing, MI 48824, USA}
\email{\href{mailto:contact@akosnagy.com}{contact@akosnagy.com}}
\urladdr{\href{http://akosnagy.com/}{akosnagy.com}}

\calclayout
\hypersetup{
	unicode=false,
    pdftoolbar=false,															
    pdfmenubar=false,															
    pdffitwindow=true,															
    pdfstartview={FitH},														
    pdftitle={Irreducible Ginzburg--Landau fields in dimension 2},				
    pdfauthor={\'Akos Nagy},													
    pdfkeywords={Ginzburg--Landau equations, gauge theory, superconductivity},	
    pdfnewwindow=true,															
    colorlinks=true,															
    urlcolor=blue,
    citecolor=black,
    linkcolor=black
}

\begin{document}

\begin{abstract}
Ginzburg--Landau fields are the solutions of the Ginzburg--Landau equations which depend on two positive parameters, $\alpha$ and $\beta$.  We give conditions on $\alpha$ and $\beta$ for the existence of irreducible solutions of these equations.  Our results hold for arbitrary compact, oriented, Riemannian 2-manifolds (for example, bounded domains in $\mathbb{R}^2$, spheres, tori, etc.) with de Gennes--Neumann boundary conditions.  We also prove that, for each such manifold and all positive $\alpha$ and $\beta$, the Ginzburg--Landau free energy is a Palais--Smale function on the space of gauge equivalence classes, Ginzburg--Landau fields exist for only a finite set of energy values, and the moduli space of Ginzburg--Landau fields is compact.
\end{abstract}

\maketitle

\section{Introduction}
\label{sec:intro}

Ginzburg--Landau theory is a phenomenological model for superconductivity which gives variational equations for an Abelian gauge field and a complex scalar field.  The gauge field is the electromagnetic vector potential, while the scalar field can be interpreted as the wave function of the so-called Bardeen--Cooper--Schrieffer ground state (a single quantum state occupied by a large number of Cooper pairs);  the norm of the scalar field is the order parameter of the superconducting phase.  Thus Ginzburg--Landau fields with non-vanishing scalar field describe superconducting phases, while vanishing scalar field corresponds to the normal phases of the material.

This paper investigates the existence --- and non-existence --- of superconducting phases in the 2-dimensional Ginzburg--Landau theory.  The topic has a vast literature in the case where the background is a planar (flat) domain in $\rl^2$ with various boundary conditions;  we recommend \cite{BBH94} for references.

\smallskip

Throughout this paper $\Sigma$ denotes a compact, oriented, Riemannian 2-manifold, which can have a non-empty, smooth boundary $\partial \Sigma$.  The Riemannian metric and the orientation together define a orthogonal complex structure $j$ and a symplectic form $\omega$.  These structures together make $\Sigma$ a K\"ahler manifold.  The volume form of Riemannian metric is $\omega$.  Let $L \rightarrow \Sigma$ be a smooth, complex line bundle with hermitian metric $h$.  Fix a unitary curvature tensor $F_0$ with finite $L^2$-norm.  We call $F_0$ the {\em external magnetic field}.  Finally fix two positive coupling constants $\alpha, \beta \in \rl_+$.  For each smooth unitary connection $\nabla$ and smooth section $\upphi$ consider the {\em Ginzburg--Landau free energy}:
\begin{equation}
\E_{\alpha, \beta}^{F_0} \left( \nabla, \upphi \right) = \tfrac{1}{2} \int\limits_\Sigma \left( \left| F_\nabla - F_0 \right|^2 + \left| \nabla \upphi \right|^2 - \alpha |\upphi|^2 + \tfrac{\beta}{2} |\upphi|^4 \right) \omega.  \label{eq:glf}
\end{equation}
Physically, if $\nabla^0$ is a unitary connection that satisfies
\begin{equation}
F_{\nabla^0} = F_0,  \label{eq:normal}
\end{equation}
then the \hyperref[eq:glf]{energy \eqref{eq:glf}} is the energy difference between the states described by $\left( \nabla, \upphi \right)$ and $\left( \nabla^0, 0 \right)$.

The variational equations of the \hyperref[eq:glf]{energy \eqref{eq:glf}} --- called the Ginzburg--Landau equations --- are gauge invariant, non-linear, second order partial differential equations.  If a solution $\left( \nabla, \upphi \right)$ is twice (weakly) differentiable, then it satisfies the de Gennes--Neumann boundary conditions.  These boundary conditions describe superconductor-insulator interfaces; see in \Cref{sec:glr} for details.

Each pair $\left( \nabla^0, 0 \right)$ that satisfies \cref{eq:normal} is a solution of the Ginzburg--Landau equations.  We call such a pair a {\em normal phase solution}, because the order parameter vanishes identically.  The \hyperref[eq:glf]{energy \eqref{eq:glf}} of any normal phase solution is zero.  Note that another pair $\left( \nabla, 0 \right)$ is normal phase solution if and only if the 1-form $a = \nabla - \nabla^0$ is closed.  In Abelian gauge theories, a pair $\left( \nabla, \upphi \right)$ is called {\em reducible} if $\upphi$ is identically zero, and {\em irreducible} otherwise.  It is easy to see that a solution of the Ginzburg--Landau equations is reducible if and only if it is a normal phase solution.

\smallskip

For the rest of the paper let
\begin{equation}
\lambda_1 = \inf \left\{ \int_\Sigma |\nabla^0 \upphi|^2 \omega \ \middle| \ F_{\nabla^0} = F_0 \ \mbox{and} \ \int_\Sigma |\upphi|^2 \omega = 1 \right\}.  \label{eq:lambda1}
\end{equation}

This non-negative real number depends on the geometric data $\left( \Sigma, j, \omega, L, h \right)$ and the curvature 2-form $F_0$, but independent of the coupling constants, $\alpha$ and $\beta$.  The quantity $\lambda_1$ plays a key role in the main theorems of this paper, which are listed below:

\smallskip

\begin{main}{{\rm [Existence]}}
\hypertarget{main:exi}
The Ginzburg--Landau equations with de Gennes--Neumann boundary conditions admit irreducible solutions if
\begin{equation}
\alpha > \lambda_1.  \label[ineq]{ineq:cond1}
\end{equation}
Moreover, if \cref{ineq:cond1} holds, then the absolute minimizers of the \hyperref[eq:glf]{energy \eqref{eq:glf}} are irreducible.
\end{main}

\smallskip

\begin{main}
\hypertarget{main:nonexi}{{\rm [Non-existence]}}
If the magnitude $|F_0| = B_0$ of the external magnetic field is constant, then $\lambda_1 = B_0$ and the Ginzburg--Landau equations with de Gennes--Neumann boundary conditions do not admit irreducible solutions if
\begin{equation}
\max \left\{ \alpha, \tfrac{\alpha}{2 \beta} \right\} \leqslant \lambda_1.  \label[ineq]{ineq:cond2}
\end{equation}

Furthermore, if $\partial \Sigma = \emptyset$ and the degree of $L$ is $d = c_1 (L) [\Sigma] \in \Z$, then
\begin{equation}
\lambda_1 = \tfrac{2 \pi |d|}{\Ar \left( \Sigma \right)}.
\end{equation}
\end{main}

\smallskip

Superconductors with $\beta > \tfrac{1}{2}$ are called {\em Type II}.  \hyperlink{main:exi}{Main~Theorem~1} and \hyperlink{main:nonexi}{Main~Theorem~2} imply the following for Type II superconductors:

\begin{main}{{\rm[The Type II case]}}
\hypertarget{main:typeii}
If $\beta > \tfrac{1}{2}$ and the magnitude $|F_0| = B_0$ of the external magnetic field is constant, then the Ginzburg--Landau equations with de Gennes--Neumann boundary conditions admit irreducible solutions if and only if
\begin{equation}
\alpha > B_0.  \label[ineq]{ineq:cond3}
\end{equation}
In fact, if \eqref{ineq:cond3} holds, then absolute minimizers of the \hyperref[eq:glf]{energy \eqref{eq:glf}} are irreducible.

When $\partial \Sigma = \emptyset$, then \cref{ineq:cond3} is equivalent to
\begin{equation}
\alpha > \tfrac{2 \pi |d|}{\Ar \left( \Sigma \right)}.  \label[ineq]{ineq:cond4}
\end{equation}
\end{main}

The results of \hyperlink{main:typeii}{Main~Theorem~3} also hold in the borderline $\beta = \tfrac{1}{2}$ case, which was proved in \cite{B90}*{Section~4} by Bradlow.  In fact, \hyperlink{main:typeii}{Main~Theorem~3} generalizes the 2-dimensional case of Bradlow's result to all parameters $\beta \geqslant \tfrac{1}{2}$.  \Cref{ineq:cond4} in \hyperlink{main:typeii}{Main~Theorem~3} provides the following phase diagram for Type II superconductors on closed 2-manifolds with constant external magnetic field:

\begin{figure}[ht!]
	\centering
		\labellist
			\small\hair 2pt
			\pinlabel $\alpha$ at 44 80
			\pinlabel $\alpha=\tfrac{\tau}{2}$ at 40 25.5
			\pinlabel \mbox{\parbox{170\unitlength}{ In the critical $\beta=\tfrac{1}{2}$ case, for each value of $\alpha$, the absolute minimizers on this line are called $\tau$-vortices, where $\tau = 2 \alpha$; see \cite{B90}.}} at 195 30
			\pinlabel $\mbox{A}$ at 68 18
			\pinlabel $\mbox{B}$ at 105 55
			\pinlabel $\Ar$ at 145 0
		\endlabellist
	\includegraphics[scale=1.6]{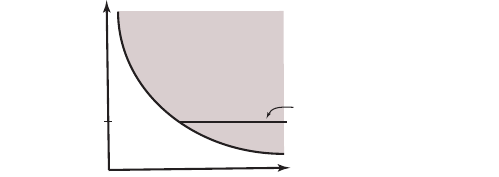}
	\captionof{figure}{{\bf Phase diagram for the $\beta \geqslant \tfrac{1}{2}$ case:} In Region~A, defined by $\alpha \leqslant \lambda_1$, only reducible solutions exist (Normal Phase).  In the complementary Region~B (shaded) there exist irreducible solutions (Superconductor Phase).}
\end{figure}

We remark, that on surfaces with boundary there is a similar phase diagram, since $2 \pi d$, in \cref{ineq:cond4} can be replaced by the magnetic flux of the external magnetic field, $\upphi_0$; see \cref{eq:flux}.

\smallskip

Our last main theorem uses a recent result of Feehan and Maridakis about the \L{}ojasiewicz--Simon inequality for analytic functions on Banach spaces \cite{FM15a}.  We show that solutions for the Ginzburg--Landau \cref{eq:gl1,eq:gl2} with de Gennes--Neumann boundary conditions exist only at finitely many energies, and the moduli space of Ginzburg--Landau fields, that is the quotient of the set of critical points of \hyperref[eq:glf]{energy \eqref{eq:glf}} by the action of the gauge group $\G$ (see definition in \Cref{sec:glr}), is compact.

\begin{main}{{\rm [Compactness]}}
\hypertarget{main:fin}
The Ginzburg--Landau free energy \eqref{eq:glf} has finitely many critical values.  Furthermore, the moduli space of Ginzburg--Landau fields
\begin{equation}
\M_{\alpha, \beta}^{F_0} = \left\{ \left( \nabla, \upphi \right) \ \middle| \ \left( \nabla, \upphi \right) \mbox{ solves \cref{eq:wgl1,eq:wgl2}} \right\} / \G.  \label{eq:glmod}
\end{equation}
is compact.
\end{main}

\smallskip

Finally, we remark that there are similar results to \hyperlink{main:exi}{Main~Theorem~1} and \hyperlink{main:nonexi}{Main~Theorem~2}, in the context of Abrikosov lattices, by Sigal and Tzaneteas; see \cites{ST12,S13}.  Abrikosov lattices are $\Z^2$-(gauge)-periodic solutions to the Ginzburg--Landau equations on the euclidean plane, $\rl^2$.  Here $\rl^2$ is viewed as the universal cover of the flat torus, and $\Z^2$ acts via deck transformations.   Furthermore, it came to our attention during the preparation of this paper that Chouchkov et al. has generalized these results to surfaces of higher genus with hyperbolic metrics \cite{CERS17}.

\smallskip

This paper is organized as follows.  \Cref{sec:glr} is a brief introduction to the Ginzburg--Landau theory on compact surfaces.  In \Cref{sec:exi} we establish \hyperlink{main:exi}{Main~Theorem~1} by proving that \hyperref[eq:glf]{Ginzburg--Landau free energy \eqref{eq:glf}} is a Palais--Smale function on the space of gauge equivalence classes of fields.  In \Cref{sec:bounds} we prove a technical lemma about the solutions of the Ginzburg--Landau equations; this is used in \Cref{sec:pr2} to prove \hyperlink{main:nonexi}{Main~Theorem~2}.  In \Cref{sec:cpt} we combine results from the previous sections and a theorem of Feehan and Maridakis to prove \hyperlink{main:fin}{Main~Theorem~4}.

\smallskip

\subsection*{Acknowledgment}  I wish to thank my advisor, Tom Parker, for his advice during the preparation of this paper.  I greatly benefited from the discussions with Paul Feehan and Manos Maridakis about the \L{}ojasiewicz--Simon inequality.  I am also grateful for the help of Benoit Charbonneau.

\smallskip

\section{Ginzburg--Landau equations on compact surfaces}
\label{sec:glr}

As is standard in gauge theory we work with the Sobolev $L_k^p$-completions of connection and fields.  For the rest of this paper we fix a connection $\nabla^0$ that satisfies \cref{eq:normal}.  The Sobolev norms are defined via the Levi-Civita connection of $\Sigma$ and the connection $\nabla^0$.  The norm of any $L_k^p$ space are denoted by $\|.\|_{k,p}$.  Furthermore $\|.\|_p$ stands for $\|.\|_{0,p}$, and $\langle.|.\rangle$ stands for the real $L^2$ inner products.  In dimension 2, $L_1^2$ embeds in $L^4$.  The weakest Sobolev norm in which the \hyperref[eq:glf]{energy \eqref{eq:glf}} is $C^1$ (in fact analytic) is $L_1^2$.  Thus let $\C_L$ be the $L_1^2$-closure of the affine space of smooth unitary connections on $L$ and $\Omega^0_L$ be the $L_1^2$-closure of the vector space of smooth sections of $L$.  Similarly, let $\Omega^k$ and $\Omega^k_L$ be the $L_1^2$-closure of $k$-forms and $L$-valued $k$-forms, respectively.  Note that $\C_L$ is now an affine space over $i \Omega^1$.  The configuration space $\C_L \oplus \Omega^0_L$ is an affine space over $i \Omega^1 \oplus \Omega^0_L$.  The gauge group $\G$ is the $L_2^2$-closure of $\Aut \left( L \right)$ in the $L_2^2$-topology.  The gauge group is canonically isomorphic to the infinite dimensional Abelian Lie group $L_2^2 \left( \Sigma, \U \left( 1 \right) \right)$ whose Lie algebra is $L_2^2 \left( \Sigma; i \rl \right)$.  Elements $g \in \G$ act on smooth pairs $\left( \nabla, \upphi \right) \in \C_L \oplus \Omega^0_L$ via
\begin{equation}
g \left( \nabla, \upphi \right) = \left( g \circ \nabla \circ g^{-1}, g \upphi \right) = \left( \nabla + g  d g^{-1}, g \upphi \right),
\end{equation}
which defines a smooth action of $\G$ on $\C_L \oplus \Omega^0_L$.

\smallskip

The \hyperref[eq:glf]{energy \eqref{eq:glf}} extends to a smooth function on $\C_L \oplus \Omega^0_L$ which now can be written as
\begin{equation}
\E_{\alpha, \beta}^{F_0} \left( \nabla, \upphi \right) = \tfrac{1}{2} \| F_\nabla - F_0 \|_2^2 + \tfrac{1}{2} \| \nabla \upphi \|_2^2 - \tfrac{\alpha}{2} \| \upphi \|_2^2 + \tfrac{\beta}{4} \| \upphi \|_4^4.
\end{equation}
Each critical point $\left( \nabla, \upphi \right) \in \C_L \oplus \Omega^0_L$ of the \hyperref[eq:glf]{energy \eqref{eq:glf}} satisfies the following equations:
\begin{subequations}
\begin{align}
\langle db | F_\nabla - F_0 \rangle + \langle b | i \: \Im \left( h \left( \upphi, \nabla \upphi \right) \right) \rangle &= 0 \qquad \forall b \in i \Omega^1  \label{eq:wgl1}  \\
\langle \nabla \uppsi | \nabla \upphi \rangle + \langle \uppsi | - \alpha \upphi + \beta |\upphi|^2 \upphi \rangle &= 0 \qquad \forall \uppsi \in \Omega^0_L.  \label{eq:wgl2}
\end{align}
\end{subequations}
If the pair $\left( \nabla, \upphi \right)$ is in $L_2^2$, then \cref{eq:wgl1,eq:wgl2} are equivalent to the {\em Ginzburg--Landau equations}:
\begin{subequations}
\begin{align}
d^* \left( F_\nabla - F_0 \right) + i \: \Im \left( h \left( \upphi, \nabla \upphi \right) \right) &= 0  \label{eq:gl1}  \\
\nabla^* \nabla \upphi - \alpha \upphi + \beta |\upphi|^2 \upphi &= 0,  \label{eq:gl2}
\end{align}
\end{subequations}
with the boundary conditions (see \cite{E98}*{page~345}):
\begin{subequations}
\begin{align}
F_\nabla &= F_0  \qquad \mbox{everywhere on }\partial \Sigma,  \label{eq:neum1}  \\
\nabla_n \upphi &= 0  \qquad \:\: \forall n \perp \partial \Sigma.  \label{eq:neum2}
\end{align}
\end{subequations}
If one writes $\nabla - \nabla^0 = a \in i \Omega^1$, then \cref{eq:neum1} becomes
\begin{equation}
da = 0  \qquad \mbox{everywhere on }\partial \Sigma.  \label{eq:newneum1}
\end{equation}

De Gennes showed that conditions \eqref{eq:neum1} and \eqref{eq:neum2} describe an interface with superconducting material on one side and insulator or vacuum on the other (see \cite{dG99}*{page~229}).  \Cref{eq:neum1} corresponds to the continuity of the magnetic field on the boundary.  The 1-form
\begin{equation}
j_{\nabla, \upphi} = \Im \left( h \left( \upphi, \nabla \upphi \right) \right)  \label{eq:supercurrent}
\end{equation}
is the {\em supercurrent}, that is the current of superconducting Cooper pairs, thus \cref{eq:neum1} means that no supercurrent is leaving or entering the surface.   Since in mathematics \cref{eq:neum2,eq:newneum1} are also called the Neumann boundary conditions, we refer to them as the {\em de Gennes--Neumann boundary conditions}.  

\smallskip

Chern--Weil theory and \cref{eq:neum1} implies that
\begin{equation}
\upphi_0 = \int\limits_\Sigma F_0 = \int\limits_\Sigma F_\nabla.  \label{eq:flux}
\end{equation}
The interpretation of \cref{eq:flux} is that the {\em magnetic flux} through the surface, $\Sigma$, is fixed.

\smallskip

By \cite{JT80}*{Theorem~2.4} every critical point is gauge equivalent to a smooth one, which in turn is a solution of the \cref{eq:gl1,eq:gl2}.  Thus gauge equivalence classes of critical points of the \hyperref[eq:glf]{energy \eqref{eq:glf}} are in one-to-one correspondence with gauge equivalence classes of smooth solutions of the \cref{eq:gl1,eq:gl2}.

\smallskip

\section{Existence}
\label{sec:exi}

To prove the existence of irreducible solutions of \cref{eq:gl1,eq:gl2} when $\alpha > \lambda_1$ we show two things: (i) absolute minimizers of the \hyperref[eq:glf]{energy \eqref{eq:glf}} exist for all $\alpha, \beta \in \rl_+$, and (ii) if $\alpha > \lambda_1$, then reducible (normal phase) solutions are not absolute minimizers.  It follows that the absolute minimizers are irreducible.

\smallskip

We say that a function $\E$ on a Banach space satisfies the {\em Palais--Smale Compactness Property} if every sequence, on which $\E$ is bounded and $D \E$ (the derivative of $\E$) converges to zero in the dual of the Banach space, has a convergent subsequence.  The next lemma shows that the \hyperref[eq:glf]{energy \eqref{eq:glf}} satisfies a gauged version of the Palais--Smale Compactness Property.

\begin{lem}
\label{lem:gps}
Let $\{ \left( a_k, \upphi_k \right) \}_{k \in \N}$ be a sequence in $i \Omega^1 \oplus \Omega^0_L$ such that $\{ \E_{\alpha, \beta}^{F_0} \left( \nabla^0 + a_k, \upphi_k \right) \}_{k \in \N}$ is bounded and $\{ D\E_{\alpha, \beta}^{F_0} \left( \nabla^0 + a_k, \upphi_k \right) \}_{k \in \N}$ converges to zero in the dual of $i \Omega^1 \oplus \Omega^0_L$.  Then there are sequences of natural numbers, $\{ k_l \}_{l \in \N}$, and gauge transformations, $\{ g_l \in \G \}_{l \in \N}$, such that $\{ g_l \left( \nabla^0 + a_{k_l}, \upphi_{k_l} \right) \}_{l \in \N}$ is convergent in $\C_L \oplus \Omega^0_L$.  
\end{lem}

\begin{proof}
Since smooth pairs in $i \Omega^1 \oplus \Omega^0_L$ are dense and the \hyperref[eq:glf]{energy \eqref{eq:glf}} is a continuous (in fact an analytic) function on $\C_L \oplus \Omega^0_L$, it is enough to prove the statement for smooth sequences.

One can complete the square in the last two terms of the \hyperref[eq:glf]{energy \eqref{eq:glf}} to get a lower bound in terms of the coupling constants and the area:
\begin{align}
\E_{\alpha, \beta}^{F_0} \left( \nabla^0 + a_k, \upphi_k \right) &= \tfrac{1}{2} \int\limits_\Sigma \left( \left| F_{\nabla^0 + a_k} - F_0 \right|^2 + \left| \left( \nabla^0 + a_k \right) \upphi_k \right|^2 + \tfrac{\beta}{2} \left( \tfrac{\alpha}{\beta} - |\upphi_k|^2 \right)^2 - \tfrac{\alpha^2}{2 \beta} \right) \omega  \\
&= \tfrac{1}{2} \| F_{\nabla^0 + a_k} - F_0 \|^2_2 + \tfrac{1}{2} \| \left( \nabla^0 + a_k \right) \upphi_k \|^2_2 + \tfrac{\beta}{4} \| \tfrac{\alpha}{\beta} - |\upphi_k|^2 \|^2_2 - \tfrac{\alpha^2}{4 \beta} \Ar \left( \Sigma \right).   \label{eq:glf2}
\end{align}
The right hand side of \eqref{eq:glf2} is greater than or equal to the constant $- \tfrac{\alpha^2}{2 \beta} \Ar \left( \Sigma \right)$.  Since all terms are positive or constant, they are bounded individually.  \Cref{eq:glf2}, and the hypotheses of the theorem give us the following:

\smallskip

\begin{enumerate}
  \item[(1)]  Since $\{ F_{\nabla^0 + a_k} \}_{k \in \N}$ is bounded in $L^2$ (by \cref{eq:glf2} and the hypotheses of the theorem) the $\{ \nabla^0 + a_k \}_{k \in \N}$ is has a subsequence, which is gauge equivalent to a weakly convergent sequence in $\C_L$, by \cite{U82}*{Theorem~3.6}.  Thus, after replacing the original sequence with a gauge equivalent one, and then taking a subsequence, we can assume that $\{ a_k \}_{k \in \N}$ is weakly convergent and hence bounded in $i \Omega^1$.  We can still assume smoothness, due to the density of smooth fields and gauge transformations.  By the Sobolev inequality, $\{ a_k \}_{k \in \N}$ is also bounded in $L^4$.
  \item[(2)]  We can also require the Coulomb gauge fixing condition to hold (for a subsequence), that is
\begin{equation}
d^* a_k = 0  \quad  \&  \quad  a (n) = 0  \quad  \forall n \perp \partial \Sigma,  \label{eq:coulomb}
\end{equation}
by to the following argument:  For all $k \in \N$ let $f_k \in \Omega^0$ be the unique solution of the equations
\begin{align}
\langle df | d f_k \rangle &= \langle f | i d^* a_k \rangle  \qquad  \forall f \in \Omega^0,  \label{eq:poi}  \\
df_k (n) &= 0  \forall n \perp \partial \Sigma.
\end{align}
Since \cref{eq:poi} is the weak formulation of a Poisson-type equation with Neumann boundary conditions, solutions exist and unique, up to additive constants; cf. \cite{W04}*{Chapter~1}.  As $\{ d^* a_k \}_{k \in \N}$ is smooth and bounded in $L^2$, by elliptic regularity $\{ f_k \}_{k \in \N}$ is smooth and bounded in $L_2^2$, and thus has a weakly convergent subsequence, $\{ f_{k_l} \}_{l \in \N}$, in $L_2^2$.  Setting $g_l = \exp \left( i f_{k_l} \right)$ for all $l \in \N$ and replacing $\{ \nabla^0 + a_k \}_{k \in \N}$ with $\{ g_l \left( \nabla^0 + a_{k_l} \right) = \nabla^0 + a_{k_l} + i df_{k_l} \}_{l \in \N}$ gives us \cref{eq:coulomb}, while not ruining the weak convergence and the smoothness of the sequence.
  \item[(3)] The integrals $\int\limits_\Sigma \left( \tfrac{\alpha}{\beta} - |\upphi_k|^2 \right)^2 \omega$ are uniformly bounded.  Since $\Sigma$ has finite volume, Jensen's inequality implies that
  \begin{equation}
  \left( \tfrac{\alpha}{\beta} \Ar \left( \Sigma \right) - \| \upphi_k \|^2_2 \right)^2 \leqslant \Ar \left( \Sigma \right) \int\limits_\Sigma \left( \tfrac{\alpha}{\beta} - |\upphi_k|^2 \right)^2 \omega.
  \end{equation}
  Thus $\{ \upphi_k \}_{k \in \N}$ is bounded in $L^2$.  But then $\{ \upphi_k \}_{k \in \N}$ is bounded in $L^4$, because
  \begin{equation}
  \| \upphi_k \|^4_4 = \int\limits_\Sigma \left( \tfrac{\alpha}{\beta} - |\upphi_k|^2 \right)^2 \omega -  \tfrac{\alpha^2}{\beta^2} \Ar \left( \Sigma \right) + 2 \| \upphi_k \|^2_2.
  \end{equation}
  \item[(4)] Combining (1), (3), and H\"older's inequality shows that $\{ \|a_k \upphi_k \|_2 \}_{k \in \N}$ is bounded.
  \item[(5)] Combining (4) and the fact that $\{ \| \left( \nabla^0 + a_k \right) \upphi_k \|_2 \}_{k \in \N}$ is bounded (by \cref{eq:glf2} and the hypotheses of the theorem) we see that $\{ \nabla^0 \upphi_k \}_{k \in \N}$ is bounded in $L^2$.
  \item[(6)] Combining (1) and (5) proves that $\{ \left( a_k, \upphi_k \right) \}_{k \in \N}$ is bounded in $i \Omega^1 \oplus \Omega^0_L$.  The embedding $L_1^2 \hookrightarrow L^p$ is completely continuous for all $p \geqslant 2$; cf. \cite{JT80}*{Proposition~2.6}.  Thus $\{ \left( a_k, \upphi_k \right) \}_{k \in \N}$ is also bounded in $L^p$ for each $p \geqslant 2$.
\end{enumerate}

\smallskip

Let $\nabla^{\rm LC} $ be the Levi-Civita connection acting on elements of $i \Omega^1$.  Since $i \Omega^1 \oplus \Omega^0_L$ is Hilbert space, there is a sequence, $\{ \left( b_k, \uppsi_k \right) \}_{k \in \N}$, such that $D\E_{\alpha, \beta}^{F_0} \left( \nabla^0 + a_k, \upphi_k \right)$ is the $L_1^2$-dual of $\left( b_k, \uppsi_k \right)$, that is for every $\left( b, \uppsi \right) \in i \Omega^1 \oplus \Omega^0_L$
\begin{equation}
D\E_{\alpha, \beta}^{F_0} \left( \nabla^0 + a_k, \upphi_k \right) \left( b, \uppsi \right) = \langle \nabla^{\rm LC} b | \nabla^{\rm LC} b_k \rangle + \langle b | b_k \rangle + \langle \nabla^0 \uppsi | \nabla^0 \uppsi_k \rangle + \langle \uppsi | \uppsi_k \rangle.  \label{eq:deps1}
\end{equation}
and thus
\begin{equation}
\| D\E_{\alpha, \beta}^{F_0} \left( \nabla^0 + a_k, \upphi_k \right) \|_{-1,2} = \| \left( b_k, \uppsi_k \right) \|_{1,2}  \label{eq:norm}
\end{equation}
for all $k \in \N$.  By hypothesis, the left hand side of \cref{eq:norm} converges to zero, which implies that $\{ \left( b_k, \uppsi_k \right) \}_{k \in \N}$ converges to zero in $i \Omega^1 \oplus \Omega^0_L$ as well.

On the other hand, by using the definition of the derivative we get
\begin{align}
D\E_{\alpha, \beta}^{F_0} \left( \nabla^0 + a_k, \upphi_k \right) \left( b, \uppsi \right) &= \tfrac{d}{dt} \left( \E_{\alpha, \beta}^{F_0} \left( \nabla^0 + a_k + t \: b, \upphi_k + t \: \uppsi \right) \right) \big|_{t = 0}  \\
&= \langle db | da_k \rangle + \langle b | i \: \Im \left( h \left( \upphi_k, \left( \nabla^0 + a_k \right) \upphi_k \right) \right) \rangle  \\
&+ \langle \left( \nabla^0 + a_k \right) \uppsi | \left( \nabla^0 + a_k \right) \upphi_k \rangle - \alpha \langle \uppsi | \upphi_k \rangle + \beta \langle \uppsi | |\upphi_k |^2 \upphi_k \rangle.  \label{eq:deps2}
\end{align}
Combining \cref{eq:deps1,eq:deps2} gives us the equation for all $\left( b, \uppsi \right) \in i \Omega^1 \oplus \Omega^0_L$
\begin{align}
\langle \nabla^{\rm LC} b | \nabla^{\rm LC} b_k \rangle + \langle b | b_k \rangle + \langle \nabla^0 \uppsi | \nabla^0 \uppsi_k \rangle + \langle \uppsi | \uppsi_k \rangle &= \langle db | da_k \rangle  \\
&+ \langle b | i \: \Im \left( h \left( \upphi_k, \left( \nabla^0 + a_k \right) \upphi_k \right) \right) \rangle  \\
&+ \langle \left( \nabla^0 + a_k \right) \uppsi | \left( \nabla^0 + a_k \right) \upphi_k \rangle  \\
&- \alpha \langle \uppsi | \upphi_k \rangle + \beta \langle \uppsi | |\upphi_k |^2 \upphi_k \rangle,  \label{eq:deps3}
\end{align}
Since \cref{eq:deps3} is a linear, elliptic partial differential equation, and $\{ \left( a_k, \upphi_k \right) \}_{k \in \N}$ is smooth, $\{ \left( b_k, \uppsi_k \right) \}_{k \in \N}$ is also smooth, by elliptic regularity.  Let $\Delta = d^* d + d d^*$ be the Laplacian on forms and $\kappa$ is the Gauss curvature of the Riemannian metric of $\Sigma$.  By the Weitzenb\"ock identity, $\left( \nabla^{\rm LC} \right)^* \nabla^{\rm LC} = \Delta - \kappa$ on $i \Omega^1$.  Since both $\{ \left( a_k, \upphi_k \right) \}_{k \in \N}$ and $\{ \left( b_k, \uppsi_k \right) \}_{k \in \N}$ are smooth, we can integrate by parts in \cref{eq:deps3}.  Recall that $d^* a_k = 0$, and so $\Delta a_k = d^* d a_k$ for all $k \in \N$, to get the following equations for all $(b, \uppsi) \in i \Omega^1 \oplus \Omega^0_L$:
\begin{subequations}
\begin{align}
\langle b | \left( \Delta + 1 - \kappa \right) b_k - \left( \Delta + |\upphi_k|^2 \right) a_k - i \: \Im \left( h \left( \upphi_k, \nabla^0 \upphi_k \right) \right) \rangle + \partial\textnormal{-terms} &= 0,  \label{eq:bk1}  \\
\langle \uppsi | \left( \left( \nabla^0 \right)^* \nabla^0 + 1 \right) \uppsi_k - \left( \nabla^0 + a_k \right)^* \left( \nabla^0 + a_k \right) \upphi_k + \alpha \upphi_k - \beta |\upphi_k|^2 \upphi_k \rangle + \partial\textnormal{-terms} &= 0.  \label{eq:bk2}
\end{align}
\end{subequations}
Equations \cref{eq:bk1} and \eqref{eq:bk2} imply the following equations in the interior of $\Sigma$:
\begin{subequations}
\begin{align}
\Delta \left( a_k - b_k \right) &= - |\upphi_k|^2 a_k + i \: \Im \left( h \left( \nabla^0 \upphi_k, \upphi_k \right) \right) + \left( 1 - \kappa \right) b_k  \label{eq:dgl1}  \\
\left( \nabla^0 \right)^* \left( \nabla^0 \right) \left( \upphi_k - \uppsi_k \right) &= - a_k^* \left( \nabla^0 \upphi_k \right) - \left( \nabla^0 \right)^* \left( a_k \upphi_k \right) - |a_k|^2 \upphi_k + \alpha \upphi_k - \beta |\upphi_k|^2 \upphi_k + \uppsi_k.  \label{eq:dgl2}
\end{align}
\end{subequations}
Since $\Sigma$ is compact, $\kappa$ is a bounded function.  Thus observations (1)-(6) imply that the right hand sides of \cref{eq:dgl1,eq:dgl2} are bounded in $L^2$.  Since $\{ \left( a_k - b_k, \upphi_k - \uppsi_k \right) \}_{k \in \N}$ is also bounded in $L^2$, elliptic regularity implies that it is bounded in $L_2^2$ as well.  The embedding $L_2^2 \hookrightarrow L_1^2$ is completely continuous (see \cite{JT80}*{Proposition~2.6}), thus $\{ \left( a_k - b_k, \upphi_k - \uppsi_k \right) \}_{k \in \N}$ has a convergent subsequence in $i \Omega^1 \oplus \Omega^0_L$.  Since $\{ \left( b_k, \uppsi_k \right) \}_{k \in \N}$ converges to zero by the argument after \cref{eq:norm}, we conclude that $\{ \left( \nabla^0 + a_k, \upphi_k \right) \}_{k \in \N}$ has a convergent subsequence in $\C_L \oplus \Omega^0_L$.
\end{proof}

\smallskip

\begin{cor}
\label{cor:ach}
The infimum of \hyperref[eq:glf]{energy \eqref{eq:glf}} is achieved by some smooth field $\left( \nabla, \upphi \right) \in \C_L \oplus \Omega^0_L$.
\end{cor}

\begin{proof}
The \hyperref[eq:glf]{energy \eqref{eq:glf}} bounded below by \eqref{eq:glf2}, and hence its infimum is not $- \infty$.  Choose any absolute minimizer sequence $\{ \left( \nabla^0 + a_k, \upphi_k \right) \}_{k \in \N}$.  Using \cite{E74}*{Proposition~2.6}, we can assume that $\{ D \E_{\alpha, \beta}^{F_0} \left( \nabla^0 + a_k, \upphi_k \right) \}_{k \in \N}$ converges to zero in the dual of $i \Omega^1 \oplus \Omega^0_L$.  By \Cref{lem:gps}, there are sequences, $\{ k_l \}_{l \in \N}$ and $\{ g_l \}_{l \in \N}$, such that $\{ g_l \left( \nabla^0 + a_{k_l}, \upphi_{k_l} \right) \}_{l \in \N}$ converges in $\C_L \oplus \Omega^0_L$.  Since the \hyperref[eq:glf]{energy \eqref{eq:glf}} is gauge invariant and analytic, the limit is an absolute minimizer and hence a critical point of the \hyperref[eq:glf]{energy \eqref{eq:glf}}.  By \cite{JT80}*{Theorem~2.4}, all critical points are gauge equivalent to a smooth pair $\left( \nabla, \upphi \right) \in \C_L \oplus \Omega^0_L$.
\end{proof}

Finally, we show that the infimum cannot be a reducible (normal phase) solution when $\alpha > \lambda_1$.

\begin{lem}
\label{lem:nonmin}
When $\alpha > \lambda_1$, all minimizing solutions are irreducible.
\end{lem}

\begin{proof}
Since all reducible solutions satisfy $\E_{\alpha, \beta}^{F_0} = 0$, it is enough to show that the minimum of the \hyperref[eq:glf]{energy \eqref{eq:glf}} is negative.  Assume that $\alpha > \lambda_1$, and pick a positive $\epsilon$ which is less than $\alpha - \lambda_1$.  Using the definition of $\lambda_1$, \cref{eq:lambda1}, we get that there exist a normal phase solution, $\left( \nabla^0, 0 \right)$, and a non-zero section, $\upphi \in \Omega^0_L$, such that
\begin{equation}
\| \nabla^0 \upphi \|^2_2 \leqslant \left( \lambda_1 + \epsilon \right) \| \upphi \|^2_2.
\end{equation}
Thus the pair $\left( \nabla^0, t \: \upphi \right) \in \C_L \oplus \Omega^0_L$ satisfies
\begin{equation}
\E_{\alpha, \beta}^{F_0} \left( \nabla^0, t \upphi \right) = \tfrac{t^2}{2} \| \nabla^0 \upphi \|^2_2 - t^2 \tfrac{\alpha}{2} \| \upphi \|^2_2 + t^4 \tfrac{\beta}{4} \| \upphi \|^4_4 \leqslant t^2 \tfrac{1}{2} \left( \lambda_1 + \epsilon - \alpha \right) \| \upphi \|^2_2 + t^4 \tfrac{\beta}{2} \| \upphi \|^4_4,
\end{equation}
which is negative for small enough $t$ since $\lambda_1 + \epsilon - \alpha < 0$.  Thus the minimum of the \hyperref[eq:glf]{energy \eqref{eq:glf}} is negative, and the hence absolute minimizer found in \Cref{cor:ach} is irreducible.
\end{proof}

\smallskip

Together \Cref{cor:ach} and \Cref{lem:nonmin} prove \hyperlink{main:exi}{Main~Theorem~1}.

\smallskip

\begin{rem}
By the definition of $\lambda_1$, \cref{eq:lambda1}, it is easy to see (using the K\"ahler identities) that 
\begin{equation}
\lambda_1 \leqslant \sup\limits_\Sigma \left( | F_0 | \right) = \| F_0 \|_\infty.
\end{equation}
Thus \hyperlink{main:exi}{Main~Theorem~1}, a fortiori, still holds if $\lambda_1$ is replaced with the (computationally simpler) quantity $\| F_0 \|_\infty$ in \cref{ineq:cond1}.
\end{rem}

\smallskip

\section{Non-existence}
\label{sec:nonexi}

\subsection{Bounds on solutions}
\label{sec:bounds}

First we prove a lemma about critical points which generalizes a result of Taubes \cite{T80}*{Lemma~3.2}, from the $\beta = \tfrac{1}{2}$ case to all positive values of $\beta$.

\smallskip

\begin{lem}
\label{lem:bounds}
There exists a positive number $C^{F_0}_{\alpha, \beta}$ such that for all critical points, $\left( \nabla, \upphi \right) \in \C_L \oplus \Omega^0_L$, of \hyperref[eq:glf]{energy \eqref{eq:glf}} we have the following inequalities:
\begin{subequations}
\begin{align}
|\upphi| & \leqslant \sqrt{\tfrac{\alpha}{\beta}}  \label[ineq]{ineq:psibound}  \\
\| F_\nabla - F_0 \|_2 & \leqslant C^{F_0}_{\alpha, \beta},  \label[ineq]{ineq:l2fbound}
\end{align}
\end{subequations}
If equality holds at any point in \cref{ineq:psibound}, then it holds everywhere and $L$ is trivial, both $\nabla$ and $\nabla^0$ are flat, and $\nabla \upphi = 0$ everywhere.

Moreover, if $\left( \nabla, \upphi \right)$ is irreducible and $F_0$ solves Maxwell's equation (that is $\upphi \neq 0$ and $d^* F_0 = 0$), then
\begin{equation}
|F_\nabla| \leqslant \max \left\{ \alpha, \tfrac{\alpha}{2 \beta} \right\} - \tfrac{1}{2} |\upphi|^2.  \label[ineq]{ineq:fbound}
\end{equation}
\end{lem}

\begin{proof}
When $\left( \nabla, \upphi \right)$ is reducible, and hence $\upphi$ is identically zero, we get $F_{\nabla} = F_0$ by \cref{eq:gl1,eq:gl2}.  Thus \cref{ineq:psibound,ineq:l2fbound} hold.  We can therefore assume that $\left( \nabla, \upphi \right)$ is irreducible.  Because the inequalities in the lemma are gauge invariant, and every solution is gauge equivalent to a smooth solution, we assume that $\left( \nabla, \upphi \right)$ is smooth.  Set
\begin{equation}
w = \tfrac{\alpha}{2\beta} - \tfrac{1}{2} |\upphi|^2,  \qquad f = i \Lambda F_\nabla,  \quad \textnormal{and} \quad f_0 = i \Lambda F_0.
\end{equation}
Using \cref{eq:gl2} straightforward computation (see \cite{T80}*{Section~III}) provides
\begin{equation}
\left( \Delta + 2 \beta |\upphi|^2 \right) w = |\nabla \upphi|^2.  \label{eq:ellip}
\end{equation}
Since $2 \beta |\upphi|^2$ and $|\nabla \upphi|^2$ are both non-negative, the maximum principle (cf. \cite{JT80}*{Proposition~3.3}) implies that $w$ is either strictly positive, or vanishes identically, thus proving \cref{ineq:psibound} and the claim about the case of equality at point.  Hence when equality holds $\upphi$ is a nowhere zero, thus $L$ is trivial.  By \cref{eq:gl2}, $\upphi$ is also parallel with respect to $\nabla$, hence by \cref{eq:gl1,eq:gl2} $\nabla$ and $\nabla^0$ are both flat.

\smallskip

Let $\star$ be the convolution of functions.  Let $G_\upphi$ be the Green's function of the positive, elliptic operator $\Delta + |\upphi|^2$ on $\Omega^0$ with Neumann boundary conditions, and similarly, $G_0$ be a Green's function of the scalar Laplacian, $\Delta$, on $\Omega^0$ with Neumann boundary conditions, which can be chosen to be everywhere positive.  Let $y \in \Sigma$ any, then for all $x \neq y$
\begin{equation}
\left( \Delta_x + |\upphi|^2 (x) \right) \left( G_0 (x,y) - G_\upphi (x,y) \right) = |\upphi|^2 (x) G_0 (x,y) \geqslant 0.
\end{equation}
Thus by the maximum principle 
\begin{equation}
G_0 (x,y) \geqslant G_\upphi (x,y).  \label[ineq]{ineq:max}
\end{equation}

Straightforward computation (see \cite{T80}*{Section~III}) provides
\begin{equation}
\left( \Delta + |\upphi|^2 \right) (f - f_0) = |\nabla^{1,0} \upphi|^2 - |\nabla^{0,1} \upphi|^2 - |\upphi|^2 f_0,
\end{equation}
and thus
\begin{equation}
\left| \left( \Delta + |\upphi|^2 \right) (f - f_0) \right| \leqslant |\nabla^{1,0} \upphi|^2 + |\nabla^{0,1} \upphi|^2 + |\upphi|^2 | f_0 | \leqslant \left( \Delta + |\upphi|^2 \right) w + |2 \beta - 1| w + \tfrac{\alpha}{\beta} |F_0|.  \label[ineq]{ineq:ellip}
\end{equation}

Both $f - f_0$ and $w$ satisfy the Neumann boundary conditions by \cref{eq:neum2,eq:gl1}.  Hence
\begin{equation}
f - f_0 = G_\upphi \star \left( |\nabla^{1,0} \upphi|^2 - |\nabla^{0,1} \upphi|^2 - |\upphi|^2 f_0 \right).  \label{eq:green}
\end{equation}
Furthermore recall that $0 \leqslant w \leqslant \tfrac{\alpha}{\beta}$.  Thus by \cref{ineq:ellip,ineq:max,eq:green} we get that
\begin{align}
\|f - f_0\|_2 & \leqslant \| G_\upphi \star \left| \left( \Delta + |\upphi|^2 \right) (f - f_0) \right| \|_2  \\
& \leqslant \| G_\upphi \star \left( |\nabla \upphi|^2 + |\upphi|^2 |F_0| \right)  \|_2  \\
& \leqslant \| G_0 \star \left( \Delta w + 2 \beta |\upphi|^2 w + \tfrac{\alpha}{\beta} |F_0| \right) \|_2  \\
& \leqslant \tfrac{\alpha}{\beta} \left( \tfrac{1}{2} + \left( \tfrac{\alpha}{4} + \| F_0 \|_2 \right) \| G_0 \| \right).
\end{align}
Setting $C^{F_0}_{\alpha, \beta} = \tfrac{\alpha}{\beta} \left( 1 + \left( |2\beta - 1| + \| F_0 \|_2 \right) \| G_0 \| \right)$ proves \cref{ineq:l2fbound}.

\smallskip

Finally, assume that $\left( \nabla, \upphi \right)$ is irreducible and $F_0$ solves Maxwell's equation.  Thus $f_0$ is constant, and \cref{eq:ellip,eq:green} and the fact that $G_\upphi \star |\upphi|^2 = 1$ give us:
\begin{align}
|f| &= | G_\upphi \star \left( |\nabla^{1,0} \upphi|^2 - |\nabla^{0,1} \upphi|^2 \right) |  \\
& \leqslant G_\upphi \star \left( |\nabla \upphi|^2 \right)  \\
&= w + (2 \beta - 1) G_\upphi \left( |\upphi|^2 w \right)  \\
& \leqslant w + \max \left\{ \tfrac{(2 \beta - 1) \alpha}{2 \beta}, 0 \right\} \left( G_\upphi \star |\upphi|^2 \right)  \\
&= w + \max \left\{ \alpha - \tfrac{\alpha}{2 \beta}, 0 \right\}  \\
&= \max \left\{ \alpha, \tfrac{\alpha}{2\beta} \right\} - \tfrac{1}{2} |\upphi|^2
\end{align}
which proves \cref{ineq:fbound}.
\end{proof}

\smallskip

\subsection{The Proof of Main Theorem 2}
\label{sec:pr2}

Assume again that $F_0$ solves Maxwell's equation, and thus $f_0$ is constant where $f_0$ was defined in \Cref{lem:bounds}.  Then the magnitude of the external magnetic field
\begin{equation}
B_0 = |F_0| = |f_0|
\end{equation}
is also constant.

\smallskip

\begin{proof}[The proof of \hyperlink{main:nonexi}{Main~Theorem~2}]
First we prove the statements about $\lambda_1$.  Let $\left( \nabla^0 \right)^{0,1}$ be the Cauchy--Riemann operator associated to the connection $\nabla^0$.   Using the K\"ahler identities and the de Gennes--Neumann boundary conditions we get that
\begin{equation}
\left( \nabla^0 \right)^* \nabla^0 = 2 \left( \left( \nabla^0 \right)^{0,1} \right)^* \left( \nabla^0 \right)^{0,1} + f_0.
\end{equation}
Hence if $f_0 \geqslant 0$
\begin{equation}
\| \nabla^0 \upphi \|^2_2 = 2 \| \left( \nabla^0 \right)^{0,1} \upphi \|^2_2 + f_0 \| \upphi \|^2_2 \geqslant B_0 \| \upphi \|^2_2,
\end{equation}
for any $\upphi \in \Omega^0_L$.  The lower bound can be achieved by choosing $\upphi$ to be holomorphic.  Such $\upphi$ exists, because $f_0$ is non-negative.  Similarly, for a negative $f_0$
\begin{equation}
\| \nabla^0 \upphi \|^2_2 = 2 \| \partial_{\nabla^0} \upphi \|^2_2 - f_0 \| \upphi \|^2_2 \geqslant |f_0| \| \upphi \|^2_2 = B_0 \| \upphi \|^2_2,
\end{equation}
and the inequality is sharp for non-zero, anti-holomorphic sections, which exist due to the sign of $f_0$.  Hence whenever $F_0$ solves Maxwell's equations $\lambda_1$ equals to $B_0$.

\smallskip

Now we prove the claim about non-existence when $\Sigma$ is closed.  Assume that $\left( \nabla, \upphi \right)$ is an irreducible solutions of \cref{eq:gl1,eq:gl2}, and $\partial \Sigma = \emptyset$.   Chern--Weil theory tells us that in this case $|F_0|$ has to be $\tfrac{2 \pi |d|}{\Ar \left( \Sigma \right)}$, where $d$ is the degree of the line bundle $L$.  By \cref{ineq:fbound} and  Chern--Weil theory again, we get that
\begin{equation}
2 \pi |d| = \left| \int\limits_\Sigma F_\nabla \right| \leqslant \int\limits_\Sigma |F_\nabla| \omega \leqslant \max \left\{ \alpha, \tfrac{\alpha}{2 \beta} \right\} \Ar \left( \Sigma \right) - \tfrac{1}{2} \| \upphi \|_2^2 < \max \left\{ \alpha, \tfrac{\alpha}{2 \beta} \right\} \Ar \left( \Sigma \right),
\end{equation}
Strict inequality holds in the last step since $\upphi \neq 0$ somewhere.  Thus there are no irreducible solutions (by the contrapositive of the previous implication) when
\begin{equation}
\max \left\{ \alpha, \tfrac{\alpha}{2 \beta} \right\} \leqslant \tfrac{2 \pi |d|}{\Ar \left( \Sigma \right)} = \lambda_1.
\end{equation}

\smallskip

When $\Sigma$ is not closed, that is $\partial \Sigma \neq \emptyset$, the magnitude of the external magnetic field, $B_0$, can be any non-negative number.  \Cref{eq:neum1,,eq:flux,ineq:fbound} gives us
\begin{equation}
0 \leqslant B_0 = \tfrac{1}{\Ar \left( \Sigma \right)} \int\limits_\Sigma F_0 = \tfrac{1}{\Ar \left( \Sigma \right)} \int\limits_\Sigma F_\nabla \leqslant \max \left\{ \alpha, \tfrac{\alpha}{2\beta} \right\} - \tfrac{1}{\Ar \left( \Sigma \right)} \tfrac{1}{2} \| \upphi \|_2^2 \leqslant \max \left\{ \alpha, \tfrac{\alpha}{2\beta} \right\}.
\end{equation}
Strict equality cannot hold in the last step if $\upphi \neq 0$ somewhere, which proves that there are no irreducible solutions when
\begin{equation}
\max \left\{ \alpha, \tfrac{\alpha}{2 \beta} \right\} \leqslant B_0 = \lambda_1.
\end{equation}
\end{proof}

\smallskip

\begin{rem}
By the definition of $\lambda_1$, \cref{eq:lambda1}, it is easy to see (using the K\"ahler identities) that 
\begin{equation}
\lambda_1 \geqslant \inf\limits_\Sigma \left( | F_0 | \right).
\end{equation}
Thus \hyperlink{main:exi}{Main~Theorem~2}, a fortiori, still holds if $\lambda_1$ is replaced with the (computationally simpler) quantity $\inf\limits_\Sigma \left( | F_0 | \right)$ in \cref{ineq:cond2}.
\end{rem}

\smallskip

\section{Compactness}
\label{sec:cpt}

A real number $E$ is a {\it critical value} of the \hyperref[eq:glf]{energy \eqref{eq:glf}} if there is a critical point, $\left( \nabla, \upphi \right) \in \C_L \oplus \Omega^0_L$, such that $\E_{\alpha, \beta}^{F_0} \left( \nabla, \upphi \right) = E$.  In this section we prove that there are only finitely many critical values and furthermore the moduli space of all Ginzburg--Landau fields is compact.
\smallskip

\begin{proof}[Proof of \hyperlink{main:fin}{Main~Theorem~4}]
First, we prove that the set of critical values is bounded.  By \eqref{eq:glf2} we see that the \hyperref[eq:glf]{energy \eqref{eq:glf}} is bounded below.  If $\left( \nabla, \upphi \right)$ is a critical point, then using \cref{eq:wgl2} with $\uppsi = \upphi$, and \cref{ineq:l2fbound} give us
\begin{align}
\E_{\alpha, \beta}^{F_0} \left( \nabla, \upphi \right) &= \tfrac{1}{2} \| F_\nabla - F_0 \|^2_2 + \tfrac{1}{2} \| \nabla \upphi \|^2_2 - \tfrac{\alpha}{2} \| \upphi \|^2_2 + \tfrac{\beta}{4} \| \upphi \|^4_4  \\
& \leqslant \tfrac{1}{2} C_{\alpha, \beta}^{F_0} + \tfrac{\alpha}{2} \| \upphi \|^2_2 - \tfrac{\beta}{2} \| \upphi \|^4_4 - \tfrac{\alpha}{2} \| \upphi \|^2_2 + \tfrac{\beta}{4} \| \upphi \|^4_4  \\
& \leqslant \tfrac{1}{2} C_{\alpha, \beta}^{F_0},
\end{align}
which proves that the critical values are bounded above.

\smallskip

Next, we show that the set critical values inherits the discrete topology from $\rl$, using a result of Feehan and Maridakis about the \L{}ojasiewicz--Simon gradient inequality.  First of all, note that the \hyperref[eq:glf]{energy \eqref{eq:glf}} is constant on gauge equivalence classes, hence we can impose the Coulomb gauge fixing condition for the proof, that is we fix a smooth critical point $\left( \nabla, \upphi \right)$ and define the {\em Coulomb slice} to be
\begin{equation}
\S_C = \left\{ \left( \nabla + a, \upphi \right) \in \C_L \oplus \Omega^0_L \middle| \: d^*a = 0 \: \& \: a(n) = 0 \: \forall n \perp \partial \Sigma \right\}. \nonumber
\end{equation}
By the argument shown in the proof of \Cref{lem:gps}, every gauge equivalence class intersects $\S_C$.  The \hyperref[eq:glf]{energy \eqref{eq:glf}} is still analytic on the affine Hilbert manifold $\S_C$, thus its Hessian at $\left( \nabla, \upphi \right)$ is a symmetric bilinear form on the tangent space of $\S_C$, and it is defined as
\begin{equation}
H \left( \left( a, \uppsi \right), \left( b, \upchi \right) \right) = \tfrac{\partial^2}{\partial s \partial t} \left( \E_{\alpha, \beta}^{F_0} \left( \nabla + s \: a + t \: b, \upphi + s \: \uppsi + t \: \upchi \right) \right) \big|_{(s,t) = (0, 0)} .  \label{eq:hess}
\end{equation}
for any $\left( a, \uppsi \right), \left( b, \upchi \right) \in T_{\left( \nabla, \upphi \right)} \S_C$.  Straightforward computation using \cref{eq:hess} shows that the operator defined by the Hessian is a compact (in fact algebraic) perturbation of a Neumann-type Laplace operator, thus a Fredholm operator of index zero;  cf. \cite{W04}*{Chapter~1}, in particular \cite{W04}*{Theorem~4.7} for details.  Hence we can use \cite{FM15a}*{Theorem~1} for the \hyperref[eq:glf]{energy \eqref{eq:glf}}:  for each critical point $\left( \nabla, \upphi \right)$ there are constants $\delta, Z > 0$ and $\theta \in [1/2,1)$, such that if $\left( \nabla', \upphi' \right) \in \S_C$ satisfies $\|\left( \nabla, \upphi \right) - \left( \nabla', \upphi' \right) \|_{1,2} < \delta$, then
\begin{equation}
|\E_{\alpha, \beta}^{F_0} \left( \nabla', \upphi' \right) - \E_{\alpha, \beta}^{F_0} \left( \nabla, \upphi \right)|^\theta \leqslant Z \| D \E_{\alpha, \beta}^{F_0} \left( \nabla', \upphi' \right) \|.  \label[ineq]{ineq:ls}
\end{equation}
In particular, if $\left( \nabla', \upphi' \right)$ is also a critical point, then $\E_{\alpha, \beta}^{F_0} \left( \nabla', \upphi' \right) = \E_{\alpha, \beta}^{F_0} \left( \nabla, \upphi \right)$.

\smallskip

Now assume that the set of critical values in infinite, and let $\{ \left( \nabla^0 + a_k, \upphi_k \right) \}_{k \in \N} \subset \S_C$ be a sequence of critical points with distinct energies.  Then $\{ \left( \nabla^0 + a_k, \upphi_k \right) \}_{k \in \N}$ satisfies the conditions of \Cref{lem:gps}, as $\{ E_k = \E_{\alpha, \beta}^{F_0} \left( \nabla^0 + a_k, \upphi_k \right) \}_{k \in \N}$ is bounded and the derivatives, $D\E \left( \nabla^0 + a_k, \upphi_k \right)$, are zero for all $k \in \N$.  Thus there are sequence, $\{ k_l \}_{l \in \N}$ and $\{ g_l \in \G \}_{l \in \N}$, such that $\{ g_l \left( \nabla^0 + a_{k_l}, \upphi_{k_l} \right) \}_{l \in \N}$ converges in $\C_L \oplus \Omega^0_L$ to a pair $\left( \nabla^0 + a_\infty, \upphi_\infty \right)$.  Since the \hyperref[eq:glf]{energy \eqref{eq:glf}} is an analytic function on $\C_L \oplus \Omega^0_L$, $\left( \nabla^0 + a_\infty, \upphi_\infty \right)$ also a critical point with energy
\begin{equation}
E = \E_{\alpha, \beta}^{F_0} \left( \nabla^0 + a_\infty, \upphi_\infty \right) = \lim\limits_{i \rightarrow \infty} E_{k_l}.
\end{equation}
Let $\delta, Z > 0$ be the constants corresponding to $\left( \nabla^0 + a_\infty, \upphi_\infty \right)$ in \cref{ineq:ls}.  By the definition of convergence there exist $i_{\mathrm{cr}}$ such that if $i > i_{\mathrm{cr}}$, then
\begin{equation}
\| \left( \nabla^0 + a_\infty, \upphi_\infty \right) - g_l \left( \nabla^0 + a_{k_l}, \upphi_{k_l} \right) \|_{1,2} < \delta.
\end{equation}
The \L{}ojasiewicz--Simon gradient inequality \eqref{ineq:ls} then implies that $E_{k_l} = E$ for all $i > i_{\mathrm{cr}}$ which contradicts our assumption that there are infinitely many critical values.

\smallskip

The compactness of the moduli space of Ginzburg--Landau fields, $\M_{\alpha, \beta}^{F_0}$, defined in \eqref{eq:glmod}, also follows from \Cref{lem:gps}:  Since the space $\left( \C_L \oplus \Omega^0_L \right) / \G$ is a metric space, it has the Bolzano--Weierstrass Property, that is a subset $\M \subset \left( \C_L \oplus \Omega^0_L \right) / \G$ is compact if only if every sequence in $\M$ has a convergent subsequence in $\M$.  By the previous observation, the \hyperref[eq:glf]{energy \eqref{eq:glf}} is bounded and has constantly zero derivative on the space of all solutions of the \cref{eq:wgl1,eq:wgl2}.  Thus every sequence in this space satisfies the conditions of \Cref{lem:gps}, thus has a subsequence that is gauge equivalent to a convergent one.  That means that the projection of the space of all solutions of \cref{eq:wgl1,eq:wgl2} to $\C_L \oplus \Omega^0_L / \G$, which is (by definition) $\M_{\alpha, \beta}^{F_0}$, has the Bolzano--Weierstrass Property, and thus compact.
\end{proof}

\bibliography{references}

\end{document}